\documentclass{ifacconf}
\usepackage{natbib}
\usepackage{amsmath} 
\usepackage{amssymb}
\usepackage{float}
\usepackage{tikz} 
\usepackage{xcolor}
\usepackage{soul}

\newcommand{\R}{\mathbb{R}}

\begin{document}
\begin{frontmatter}
{\footnotesize
© 2026 the authors. This work has been accepted to IFAC for publication
under a Creative Commons Licence CC-BY-NC-ND.
}
\title{On the behavior assignment problem}

\author[DEI]{Francesca Mazzolani} 
\author[DEI]{Michelangelo Bin}
\author[DEI]{Lorenzo Marconi}

\address[DEI]{Department of Electrical, Electronic, and Information Engineering ``Guglielmo Marconi'', University of Bologna, Bologna, Italy\\
  (e-mail: f.mazzolani@unibo.it; michelangelo.bin@unibo.it; lorenzo.marconi@unibo.it).}

\begin{abstract}
This paper introduces the asymptotic behavior assignment problem for nonlinear systems. Given a controlled system and a reference system with an ``open'' input, the goal is to design a regulator such that, for every admissible input, the asymptotic input--output behavior of the closed-loop system reproduces that of the reference. This formulation captures, as special cases, classical model matching, disturbance rejection, and master--slave synchronization, but does not assume that an explicit tracking or regulation error is available for feedback. Motivated by nonlinear output regulation, we discuss how steady-state concepts for autonomous systems must be adapted when the closed-loop dynamics is not autonomous. In a SISO normal-form setting we devise sufficient conditions for the solution of the behavior assignment problem by introducing a synchrony-detection signal whose convergence to zero is equivalent to successful behavior assignment, thereby reducing the problem to a standard stabilization one. Two examples---a tunnel-diode circuit with multiple input-dependent equilibria, and a pendulum frequency-matching problem---illustrate how the proposed framework avoids artificially selecting a specific steady state.
\end{abstract}

\begin{keyword}
Model matching; output regulation; internal model principle; synchronization; nonlinear systems.
\end{keyword}

\end{frontmatter}
\section{Introduction}
\subsection{The behavior assignment problem}

Consider a controlled system of the form
\begin{equation}
  \dot x = f(x,u,t), \qquad y = h(x,t)
  \label{plant}
\end{equation}
with state $x(t) \in \mathbb{R}^n$, control input $u(t) \in \mathbb{R}^m$, and output $y(t) \in \mathbb{R}^p$. The explicit time dependence
in~\eqref{plant} is used to model non-stationary effects or
additional exogenous inputs, such as disturbances, acting on the
controlled system. Let $v$ be an ``open'' input, belonging to a
given class $\mathcal{V}$, and consider a reference system
\begin{equation}
  \dot x^\star = f^\star(x^\star,v,t), \qquad y^\star = h^\star(x^\star,t)
  \label{reference}
\end{equation}
with state $x^\star(t) \in \mathbb{R}^\nu$ and output $y^\star(t) \in \mathbb{R}^p$.
The reference system is not necessarily physically available. In
many cases it should be thought of as a virtual model encoding the
desired long-term input--output behavior.

The controller has internal state $\eta(t)\in\mathbb{R}^{n_\eta}$ and is assumed to have the general form
\begin{equation}
  \dot \eta = c(\eta,y,v,t), \qquad
  u = \alpha(\eta,y,v,t).
  \label{controller}
\end{equation}
The interconnection of~\eqref{plant} and~\eqref{controller} yields a closed-loop system with state
$(x,\eta)$, input $v$, and output $y$.

To formalize the desired matching between~\eqref{plant} and~\eqref{reference} we work with asymptotic solutions. Consider a
generic system of the form
\begin{equation}
  \dot \chi = \gamma(\chi,w,t),
  \label{eq:chi-system}
\end{equation}
with state $\chi$ and input $w$. A function $\chi_{\mathrm{as}}$ is called an \emph{asymptotic
solution} of~\eqref{eq:chi-system}, corresponding to the input
$w$, if it is differentiable and
\[
  \lim_{t\to\infty}
  \bigl|\dot\chi_{\mathrm{as}}(t) -
        \gamma(\chi_{\mathrm{as}}(t),w(t),t)\bigr| = 0.
\]
We observe that, $\chi_{\mathrm{as}}$ may not be an exact
trajectory of~\eqref{eq:chi-system} for all times, since it satisfies the equation~\eqref{eq:chi-system} only asymptotically. As a simple example, consider the system $\dot \chi = -\chi + e^{-t}$. The function $\chi_{\mathrm{as}}(\cdot) = 0$ is an asymptotic solution, yet not a solution of the system. In our framework, asymptotic solutions play
the role of ``steady states'' for systems with inputs, and will be
used to describe the long-term behavior of both the reference and
the controlled system.

With these elements at hand, the behavior assignment problem is defined as follows: design a
regulator of the form~\eqref{controller} so that, for every
$v\in\mathcal{V}$ and every output $y$ of the closed-loop system
corresponding to $v$, there exists an asymptotic solution
$x^\star_{\mathrm{as}}$ of the reference system~\eqref{reference}, with corresponding output
$y^\star_{\mathrm{as}}$, such that
\[
  \lim_{t\to\infty} \bigl|y(t) - y^\star_{\mathrm{as}}(t)\bigr| = 0.
\]
Equivalently, for each admissible input $v$ the controller should be able to induce, in the limit, an input--output
behavior of the reference system, using only the output
$y$, and the external input $v$.\\
\textit{Remark.} For simplicity, we consider the case where the output $y$ coincides with the signal available for feedback in~\eqref{controller}. Nevertheless, the framework can be straightforwardly extended to the case in which~\eqref{controller} processes a measured output different from $y$, such as the full state $x$ (see the examples in Section~$3$).

\subsection{Discussion and relation with classical paradigms}

The aim of assigning a prescribed input--output behavior to a controlled system is closely related to several classical control paradigms, most notably nonlinear model matching, and output regulation. At the same time, the formulation adopted here differs from these paradigms in structural ways that are relevant for both problem statements and solution methods.

Nonlinear model matching typically requires that the closed-loop input--output map matches that of a given reference model. In geometric approaches such as \cite{DiBenedettoGrizzle1994}, the controlled system and the reference model are combined into an extended system, an output error is defined, and one looks for an error-feedback law rendering a suitable zero-error manifold asymptotically stable. This leads to sharp necessary and sufficient conditions, expressed in terms of invariant distributions and controlled invariants, and to an internal-model-like interpretation: on the zero-error manifold, the closed-loop dynamics embeds the target dynamics, so that an internal model of the reference is present in the controlled system--regulator combination. However, the price of this precision is an explicit reliance on the error system, and on the specific target trajectory selected by the initialization of the copy of the reference model.

The problem considered in this paper is also closely related to output
regulation \citep{SICON_2008}. If one regards the reference system~\eqref{reference}
as an exosystem and imagines that its output $y^\star$ is available,
then the natural quantity to be regulated would be the error
\[
  e(t) = y(t) - y^\star(t),
\]
and the objective of asymptotic behavior assignment could be restated
as driving $e(t)$ to zero for all admissible inputs $v$. From
this viewpoint,~\eqref{reference} plays the role of the signal
generator in the output regulation framework, and the present setting
is a variant of nonlinear regulation in the sense of
\cite{ByrnesIsidori2003}.

There are, however, two structural peculiarities that distinguish our formulation from canonical output regulation problems. First, the reference dynamics~\eqref{reference} is inherently ``open'': it is driven by the external input $v$, which is not generated by an
autonomous exosystem but belongs to a prescribed class
$\mathcal{V}$. As a consequence, one cannot simply form an autonomous
extended system by interconnecting the controlled system, the exosystem and the regulator, and
analyze its limit sets as in
\cite{ByrnesIsidori2003}. The notion of steady state
must instead be adapted to this open setting, which motivates the use
of asymptotic solutions associated with given input signals.

Second, and equally important, the error $e(t)=y(t)-y^\star(t)$ may not be
physically available in general: the reference system does not need to be implemented, and neither $x^\star$ nor $y^\star$ are required to be measured. The controller only has access to the output $y$ and to the input $v$, and any quantity playing the role of a regulation error must be constructed indirectly from these signals. This lack of a measurable error prevents a direct application of classical output regulation designs and of their internal model principles. The constructions developed in this paper can be seen as a way to recover, in this open and error-free setting, some of the structural insights of regulation theory while working directly at the level of asymptotic behaviors.

The formulation proposed in this paper takes a different starting point. The reference system~\eqref{reference} is treated as an open input--output behavior, and the desired specification is encoded by the collection of all asymptotic input--output pairs $(v,y^\star)$ that it can generate. The controller has access to the output $y$ and to the input $v$, but, in general, it may not have access to the state or output of the reference system (which could be a ``virtual system''). Asymptotic behavior assignment is achieved when, for each admissible input $v$, the corresponding closed-loop output $y$ asymptotically coincides with some asymptotic output $y^\star$ of the reference. This viewpoint is reminiscent of \cite{Willems1991} behavioral framework, in that it focuses on sets of trajectories rather than on a specific error trajectory, but here it is developed directly at the level of state-space representations, without resorting to the full behavioral machinery.

Two simple examples, discussed in detail in Section 3, help illustrate the distinctions. First, consider a tunnel-diode circuit with state $(z,\mu)$ and dynamics \begin{equation}
\label{tdsys}
  C\dot z = \mu - h(z), \qquad
  L\dot \mu = -z - R \mu + u + d,
\end{equation} where $C$ and $L$ are the capacitance and inductance, and $R$ is the resistance. The nonlinearity $h$ has an ``S'' shape, and $d$ is an external disturbance, see Example 2.1 in \cite{Khalil}. For a given constant input $v$, the system with $u+d=v$ may have up to three distinct equilibria. A disturbance rejection problem can be posed by taking as reference the disturbance-free system (with $d=0$) driven by the same constant input. An error-based solution that embeds this disturbance-free system in the controller, forms an error between the measured $\mu$ and the replica, and uses it for feedback, inevitably selects one equilibrium through its initialization: the closed loop is forced to a specific steady state. In contrast, an asymptotic behavior assignment design aims at restoring the prescribed relation between the constant input and the asymptotic output, while letting the actual equilibrium reached depend on the initial condition of the controlled system. The multiplicity of equilibria is preserved rather than artificially destroyed by the controller.

In a second example, the reference system is an undamped pendulum with a given natural frequency, while the controlled system is a pendulum with a different natural frequency. Classical model-following schemes tend to impose one particular periodic trajectory of the reference pendulum on the controlled system through an embedded model and an error feedback. In the present framework, the specification is instead formulated in terms of asymptotic input--output behavior, requiring for instance, an asymptotically matching frequency and amplitude, while allowing phase and internal motions to depend on initial conditions. These examples motivate the need for tools that work directly at the level of asymptotic behaviors.

The rest of the paper develops these ideas. Section 2 introduces the notion of synchrony detection and, in a SISO normal-form setting, shows how to construct a synchrony detector whose output vanishes if and only if the desired behavior assignment is achieved, thereby reducing the problem to a stabilization task. Section 3 presents the motivating examples, and Section 4 concludes with remarks on internal-model-like structures and on the relationship between the present viewpoint and existing internal model principles.
\section{A pathway to the solution}
\subsection{On the concept of synchrony detection}

As discussed in the previous section, most existing design frameworks rely on the availability of an explicit error variable that inevitably forces a specific steady-state trajectory. In our setting, instead, we aim to develop controllers that do not force a specific steady state. This naturally raises the question of whether it is possible to construct a \emph{proxy}---derived from the available measurements---that fulfills the same conceptual role of a regulation error without necessarily forcing a unique steady state.

In this regard, we draw inspiration from \citep{WANG2020}, where a \emph{washout filter} is used to convert an auxiliary measurement---non-vanishing in steady state---into a filtered signal that converges to zero while still retaining the dynamical information required for stabilizing the internal dynamics of the controlled system. In that setting, the washout filter is used to remove the steady-state component from the output, ensuring that the signal fed to the regulator only excites the relevant transient dynamics and does not contaminate the adaptation mechanism with non-decaying terms. The resulting filtered output acts as a “virtual error”, effectively reproducing the blocking-zero mechanism typical of pre-processing internal-model architectures \citep{bin_internal_2022}.

Motivated by this idea, while following a post-processing approach, we introduce a module named \textit{synchrony detector} that operates on the measured output of the controlled system~\eqref{plant} and generates a filtered output signal $\varepsilon$. This detector can be generically expressed as
\begin{equation}\label{generic_detector}
\begin{aligned}
\dot\sigma &= \phi(\sigma, y, v, t) \\
\varepsilon &= \gamma(\sigma, y, v, t)
\end{aligned}
\end{equation}
where $\sigma(t) \in \mathbb{R}^{n_\sigma}$ is its internal state.
This signal, which we will refer to as the \textit{synchronization error}, plays the role of an estimated error proxy and provides a measure of how closely the input--output mapping of the controlled system~\eqref{plant} matches that of the desired behavior, without necessarily using any direct measurement from the reference system~\eqref{reference}. The synchrony detector~\eqref{generic_detector} must therefore be designed so that $\varepsilon$ vanishes at steady state while still capturing the relevant input--output information of the reference system~\eqref{reference}. We formalize this requirement through the following property.

\textit{Definition 1.}\label{Synch_property}
System~\eqref{generic_detector} is said to have the synchrony detection property for the reference system~\eqref{reference} and for the class $\mathcal{V}$ if, for every $v \in \mathcal{V}$, and for every output $y$ generable by~\eqref{plant}, the condition $\lim_{t\to\infty} \varepsilon(t)=0$ implies that there exists an asymptotic solution $x_{\mathrm{as}} ^\star$ of~\eqref{reference} corresponding to $v$, with output $y_{\mathrm{as}}^\star(t) = h^\star(x_{\mathrm{as}}^\star (t),t)$, such that \begin{equation}
    \notag
    \lim_{t\to \infty} |y(t) - y^{\star}_{\mathrm{as}}(t)| =0. 
\end{equation}
\subsection{The case of SISO normal forms}
In the following analysis, we restrict our focus to single--input single--output systems (i.e., $m=p=1$), and assume that system~\eqref{plant} has a well-defined global relative degree $r \geq 1$, and a globally defined normal form. We then rewrite~\eqref{plant} in the normal coordinates as \begin{equation}
    \begin{split}
    \label{plant_nf}
        \dot z &= f_0(z,\mu,t)\\
        \dot \mu_i &= \mu_{i+1} \quad \text{with\quad} i = 1,\dots, r-1\\
        \dot \mu_r &= q(z,\mu,t) + b(z, \mu,t) u
    \end{split}
\end{equation}
with state $(z(t), \mu(t)) \in \mathbb{R}^{n_z}\times \mathbb{R}^r$, where $n_z = n - r$, and output $y(t) =\mu_1(t) \in \mathbb{R}$. Similarly, we also assume that the reference system~\eqref{reference} can be globally expressed in the following normal form
\begin{equation}
    \begin{split}
    \label{reference_nf}
        \dot z^ \star &= f_0 ^ \star (z^\star, \mu^\star,t)\\
        \dot \mu_i^\star &= \mu_{i+1}^\star \quad \text{with\quad} i = 1,\dots, r-1\\
        \dot \mu_r ^\star &= q^\star (z^\star, \mu^\star,t) +  b^\star(z^\star, \mu^\star,t)v
    \end{split}
\end{equation}
with state $(z^\star(t), \mu^\star(t)) \in \mathbb{R}^{n_{z^\star}} \times \mathbb{R}^r$, where $n_{z^\star}=\nu-r$. 
To ease the forthcoming derivations, we assume that the reference system has the same relative degree $r$ of the controlled system. Instead, we do not impose any constraint on the dimension or structure of the internal dynamics. In particular, the dimension $n_{z^\star}$ of the state $z^\star$ in~\eqref{reference_nf} is not required to coincide with the dimension $n_z$ of $z$ in~\eqref{plant_nf}.

Under these assumptions, we can devise a constructive approach to the design of the synchrony detector~\eqref{generic_detector}, which can simply be taken as a replica of the reference system~\eqref{reference_nf} with the measured output $\mu_1(t)$ of the system~\eqref{plant_nf} directly incorporated into its structure as a driving signal. In particular, a synchrony detector for the reference~\eqref{reference_nf} can be built as follows
\begin{equation}
    \begin{split} \label{filter}
        \dot \zeta &= f_0^\star (\zeta, \mu_1, \bar\xi,t)\\
        \dot \xi_{i} &=  \xi_{i+1} \quad \text{with\quad} i = 1,\dots, r-1\\
        \dot \xi_r &= q^\star(\zeta, \mu_1, \bar \xi,t) + b^\star(\zeta, \mu_1, \bar \xi ,t)v \\
        \varepsilon &= \mu_1 - \xi_1
    \end{split}
\end{equation}
with states $(\zeta (t), \xi(t)) \in \mathbb{R}^{n_{z^\star}} \times \mathbb{R}^{r}$, output $\varepsilon$, and where $\bar \xi (t) = (\xi_2(t), \dots, \xi_r(t))$.

The following theorem shows that indeed~\eqref{filter} has the detection property.
\begin{thm}\label{TH2}
Let $\mathcal{V}$ be a set of bounded functions, and assume $f_0^\star$, $q^\star$, and $b^\star$ are uniformly continuous in all their arguments. Then, system~\eqref{filter} has the detection property for the reference system~\eqref{reference_nf} and for the class $\mathcal{V}$.
\end{thm}
\begin{pf}
Pick $v \in \mathcal{V}$ and let $(z, \mu)$ be a solution of~\eqref{plant_nf}. Moreover, let $(\zeta, \xi)$ be a solution of~\eqref{filter} corresponding to $v$ and $y = \mu_1$ for some arbitrary input $u$.

Assume that $\lim_{t \to \infty} \varepsilon(t) = 0$. Then $\lim_{t \to \infty}|\mu_1(t) - \xi_1(t)| = 0$. Since $f_0^\star$, $q^\star$, and $b^\star$ are uniformly continuous, and $v$ is bounded, it follows that 
\begin{equation}\notag
        \begin{split}
            \lim_{t \to \infty} &|\dot \zeta  - f_0^\star (\zeta, \xi, t)|  =\\ =&\lim_{t \to \infty} | f_0^\star (\zeta, \mu_1, \xi_2, \dots, \xi_r,t) - f_0^\star (\zeta, \xi_1, \xi_2, \dots, \xi_r,t)| \\ =&\ 0
            \end{split}
\end{equation}
and
\begin{equation}\notag
        \begin{split}
            \lim_{t \to \infty} &|\dot \xi_r  - q^\star(\zeta, \xi,t) - b^\star(\zeta, \xi,t)v(t)|   =\\ =&\lim_{t \to \infty} |q^\star(\zeta, \mu_1, \bar \xi,t) + b^\star(\zeta, \mu_1, \bar \xi,t)v \\ &- q^\star(\zeta, \xi, t) - b^\star(\zeta, \xi,t)v| \\ =&\ 0 
        \end{split}
\end{equation}

Then, $(z^{\star}_{\mathrm{as}}(t), \mu^{\star}_{\mathrm{as}}(t)) := (\zeta(t), \xi(t))$
is an asymptotic solution of~\eqref{reference_nf} (corresponding to the previously fixed $v$) with output $y^{\star}_{\mathrm{as}} = \mu_{\mathrm{as},1} ^\star$ and $\lim_{t \to \infty} |y(t) - y^{\star}_{\mathrm{as}}(t)| = 0. $\hspace*{\fill} $\square$ 
\end{pf}
The interconnection of the synchrony detector~\eqref{filter} and the controlled system~\eqref{plant_nf}, leads to an \emph{extended system} that can be expressed as
\begin{equation} \label{extended_system}
    \begin{split}
        \dot z &= f_0(z,\mu,t)\\
        \dot \mu_i &= \mu_{i+1}, \qquad \qquad \text{with\quad} i = 1,\dots, r-1\\
        \dot \mu_r &= q(z,\mu,t) + b(z, \mu,t) u\\
        \dot \zeta &= f_0^\star (\zeta, \mu_1, \bar\xi,t)\\
        \dot \xi_{i} &=  \xi_{i+1}, \qquad \qquad \text{with\quad} i = 1,\dots, r-1\\
        \dot \xi_r &= q^\star(\zeta, \mu_1, \bar \xi,t) + b^\star(\zeta, \mu_1, \bar \xi ,t)v \\
        \varepsilon&= \mu_1- \xi_1
    \end{split}
\end{equation}
Let us define $\epsilon:=\mu-\xi$ (in particular, $\epsilon_1=\varepsilon$).
Changing coordinates from $\mu$ to $\epsilon$, we obtain the \emph{error dynamics}
\begin{equation}\label{error_system}
\begin{aligned}
\dot z =&\ f_0(z,{\epsilon + \xi}, t)\\
\dot \zeta =&\ f_0^\star(\zeta,{\varepsilon+ \xi_1},\bar \xi,t)\\
\dot \xi_i =&\ \xi_{i+1}, \qquad \qquad \text{with\quad} i = 1,\dots, r-1\\
\dot \xi_r =&\ q^\star(\zeta,{ \varepsilon+ \xi_1},\bar \xi,t) + b^\star(\zeta,{\varepsilon+ \xi_1}, \bar \xi,t)\,v\\
\dot \epsilon_i =&\ \epsilon_{i+1}, \qquad \qquad \text{with\quad} i = 1,\dots, r-1\\
\dot \epsilon_r =&\ q(z,{\epsilon + \xi},t) + b(z,{\epsilon + \xi},t)\,u
           - q^\star(\zeta,{\varepsilon+ \xi_1},\bar \xi,t)\\& - b^\star(\zeta,{\varepsilon+ \xi_1},\bar \xi,t)\,v
\end{aligned}
\end{equation}
In these terms, the behavior assignment problem thus reduces to the design of a feedback controller capable of driving the synchronization error $\varepsilon$ to zero. In turn, this is an output stabilization problem, for which numerous well-established methodologies are available in the literature, including high-gain techniques such as \cite{Freidovich2006}, or \cite{SICON_2008}, \cite{Teel_Praly_1995}, and \cite{Isidori1995}.
In particular, consider a stabilizer of the form
\begin{equation}\label{prop_contr}
\begin{split}
    \dot \eta &= \varrho(\eta,\xi, \epsilon, v)\\
    u &= \varphi(\eta,\xi, \epsilon , v)   
\end{split}
\end{equation}
with state $\eta(t) \in \R^{n_\eta}$. 
Then, the closed-loop system reads as
\begin{equation} \label{overall_reg}
    \begin{split}
        \dot z =&\ f_0 (z, \epsilon + \xi, t)\\
        \dot \zeta =&\ f_0^\star (\zeta, {\varepsilon + \xi_1, \bar \xi},t)\\
        \dot \xi_{i} =&\  \xi_{i+1}, \qquad \qquad \text{with\quad} i = 1,\dots, r-1\\
        \dot \xi_r =&\ q^\star(\zeta, {\varepsilon + \xi_1, \bar \xi} ,t) + b^\star(\zeta,{\varepsilon + \xi_1, \xi},t)v\\
        \dot \epsilon_i =& \epsilon_{i+1}, \qquad \qquad \text{ with\quad} i = 1,\dots, r-1\\
        \dot \epsilon_r =&\ q(z,{\epsilon + \xi},t) + b(z,{\epsilon + \xi},t)\,u
           - q^\star(\zeta,{\varepsilon+ \xi_1},\bar \xi,t)\\& - b^\star(\zeta,{\varepsilon+ \xi_1},\bar \xi,t)\,v\\
        \dot \eta =&\ \varrho(\eta, \xi, \epsilon, v)\\
        u =&\ \varphi(\eta, \xi, \epsilon, v)        
    \end{split}
\end{equation}
and the following result holds.
\begin{thm}
Let~\eqref{prop_contr} be such that
$\mathcal{B} = \{(z, \zeta, \xi, \epsilon, \eta) \in \R^{n_z + n_{z^\star} + 2r + n_{\eta} }: \epsilon = 0\}$
is asymptotically stable for the closed-loop system~\eqref{overall_reg}.
Then, the control system~\eqref{filter},~\eqref{prop_contr}
solves the behavior assignment problem for~\eqref{plant_nf}--\eqref{reference_nf}.
\end{thm}
\section{Examples}
\subsection{Tunnel diode with constant input}
As a first illustrative study case, we consider the tunnel-diode circuit~\eqref{tdsys} driven by a constant voltage source $u$ and subject to an external disturbance $d$. We deal with a state feedback solution, namely both $z$ and $\mu$ are measured. 
The desired behavior is described by the following reference system
\begin{equation}
\begin{aligned}
C\dot z^\star &= \mu^\star - h(z^\star) \\
L\dot \mu^\star &= -z^\star - R \mu^\star + v
\end{aligned}
\label{eq:model}
\end{equation}
in which $(z^\star, \mu^\star)$ is the state,   $y^\star = \mu^\star$ is the reference output,  and  $h$ is given by
$h(z^\star) = 17.76 \,z^\star - 103.79 \,z^\star{}^{2}
 + 229.62 \,z^\star{}^{3} - 226.312 \, z^\star{}^{4} + 83.72 \, z^\star{}^{5}
$.
The reference system is in normal form and has unitary relative degree. For $v=1.2\,\mathrm{V}$, the model admits three different equilibrium points, two of which are stable and one is unstable, revealing the bistable nature of the device.

Following the design methodology of Section 2, we introduce the synchrony detector
\begin{equation}
\begin{aligned}
C\dot \zeta &= \mu - h(\zeta)\\
L\dot \xi &= -\zeta - R \mu + v\\
\varepsilon &= \mu - \xi
\end{aligned}
\label{eq:detector_TD}
\end{equation}
According to Theorem~\ref{TH2}\footnote{We observe that, because of the form of $h(\cdot)$, the key assumption of Theorem~\ref{TH2} asking that $f_0^\star$ is uniformly continuous is not fulfilled. However, since the trajectories of~\eqref{eq:model} originated from compact sets are bounded, the assumption in question can be forced by substituting $h$ with a uniformly continuous function matching it within the set where trajectories evolve.}, the detector satisfies the detection property for the reference system \eqref{eq:model}. With $\mu = \varepsilon + \xi$, the extended error dynamics~\eqref{overall_reg} reads
\begin{equation}
\begin{aligned}
C\dot z &=\varepsilon + \xi - h(z)\\
C\dot \zeta &= \varepsilon + \xi - h(\zeta)\\
L \dot \xi &= - \zeta - R(\varepsilon + \xi)+v\\
L\dot \varepsilon &= \zeta - z +  u + d - v
\end{aligned}
\end{equation}
Since $d(t)$ in~\eqref{tdsys} is constant, a simple integrator can be incorporated in the controller to eliminate the disturbance. The resulting (state feedback) regulator is then given by
\begin{equation}
\begin{aligned}
C\dot \zeta &= \varepsilon + \xi - h(\zeta) \\
L\dot \xi &= -\zeta - R (\varepsilon + \xi) + v \\
\dot \eta &= - k_2 \varepsilon\\
u &= z - \zeta + v - k_1 \varepsilon + \eta 
\end{aligned}
\end{equation}
with gains $k_1, k_2 > 0$.
By letting $\tilde \eta = \eta + d$, it turns out that the closed-loop system is
\[
\begin{aligned}
 \dot \varepsilon &= - k_1 \varepsilon + \tilde \eta\\
 \dot {\tilde  \eta} &= - k_2 \varepsilon
\end{aligned}
\]
which has the origin asymptotically stable, driving a system with state $(z,\zeta,\xi)$, which has bounded trajectories for vanishing input $\varepsilon$. The controller thus solves the problem.
Simulations in Fig.~\ref{fig:tunnel_diode_simulation} and Fig.~\ref{fig:tunnel_diode_simulation_1} show the results for different initial conditions of the closed-loop system. Fig.~\ref{fig:tunnel_diode_simulation} illustrates the synchronization error $\varepsilon = \mu - \xi$, which converges to zero in all considered initial conditions, confirming correct reproduction of the reference behavior. Fig.~\ref{fig:tunnel_diode_simulation_1} shows the corresponding output trajectories: depending on the initial condition, the output converges to one of the two stable equilibria of the reference system. Different initial states lead to different trajectories, all satisfying the synchronization requirement.
Most notably, the closed-loop trajectory is not predetermined by the controller alone; different initial conditions of the controlled system give rise to different trajectories, all satisfying the required synchronization property. This is the key feature of the proposed approach. Rather than enforcing convergence to a single steady-state output, it assigns a {\it family} of admissible asymptotic behaviors. By preserving the dependence on the initial state, the design allows the system to retain flexibility, so that the long-term behavior reflects both the controlled system’s intrinsic dynamics and its initialization.
\begin{figure}
\begin{center}
\includegraphics[width=8.4cm, height=4.5cm]{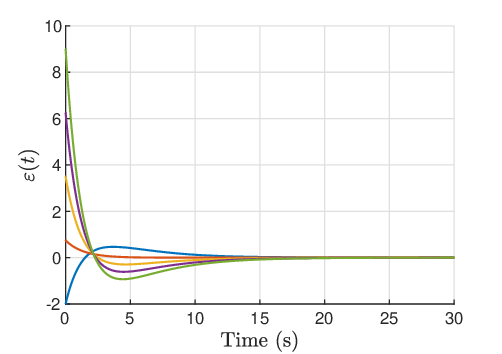}   
\caption{Synchronization error $\varepsilon(t)$ for different initial conditions of the controlled system. Results shown for $k_1 = 5$ and $k_2 = 1$.} 
\label{fig:tunnel_diode_simulation}
\end{center}
\end{figure}
\begin{figure}
\begin{center}
\includegraphics[width=8.4cm, height=4.5cm]{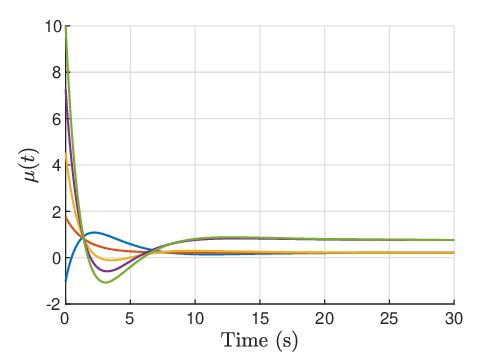}  
\caption{Controlled output $y(t)=\mu(t)$ for different initial conditions.} 
\label{fig:tunnel_diode_simulation_1}
\end{center}
\end{figure}
\subsection{Frequency-Matching}
As a second illustrative example, we consider a controlled undamped pendulum of the form
\begin{equation}
\begin{aligned}
\dot{\mu}_1 &= \mu_2\\
\dot{\mu}_2 &= -\omega_p^2 \sin(\mu_1) + u
\end{aligned}
\label{eq:pendulum_plant}
\end{equation}
where $\mu_1$ is the angle, $\mu_2$ is its derivative, $u$ is the control torque, and $\omega_p$ is the natural frequency of the pendulum. We take as output the angle by letting $y = \mu_1$. The reference system is given by
\begin{equation}
\begin{aligned}
\dot{\mu}^\star_{1} &= \mu^\star_{2}\\
\dot{\mu}^\star_{2} &= -\omega_m^2 \sin(\mu^\star_{1}) + v
\end{aligned}
\label{eq:pendulum_model}
\end{equation}
with input $v\in\mathbb{R}$, natural frequency $\omega_m \ne\omega_p$, and reference output $y^\star = \mu_{1}^\star$. 
Depending on the initial condition, the reference~\eqref{eq:pendulum_model} can exhibit bounded periodic solutions, equilibria, or continuous rotational motions.
This system is already expressed in normal form with relative degree $r=2$ between input $v$ and output $y^\star$.

A synchrony detector for the reference system (\ref{eq:pendulum_model}) and for the class of bounded inputs $\mathcal{V}$ can be designed as 
\begin{equation}
\begin{aligned}
\dot{\xi}_1 &= \xi_2\\
\dot{\xi}_2 &= -\omega_m^2 \sin(\mu_1) + v\\
\varepsilon &= \mu_1 - \xi_1
\end{aligned}
\label{eq:pendulum_detector}
\end{equation}
with state $(\xi_1,\xi_2)\in\mathbb{R}^2$. System \eqref{eq:pendulum_detector} satisfies the detection property in view of Theorem~\ref{TH2}. Defining $\epsilon_1 = \varepsilon$ and $\epsilon_2 = \mu_2 - \xi_2$, the corresponding error dynamics are given by
\begin{equation}
    \begin{aligned}
        \dot \xi_1 =&\ \xi_2\\
        \dot \xi_2 =&\ -\omega_m ^ 2 \sin (\varepsilon + \xi_1) + v \\
       \dot \epsilon_1 =&\ \epsilon_2\\
        \dot \epsilon_2 =&\ -\omega_p^2 \sin (\varepsilon + \xi_1) + u +\omega_m ^2  \sin (\varepsilon + \xi_1) -v 
    \end{aligned}
\end{equation}
To enforce asymptotic synchronization between \eqref{eq:pendulum_plant}-\eqref{eq:pendulum_model}, we select the following stabilizing law
\begin{equation}
u = (\omega_p^2 - \omega_m^2) \sin(\varepsilon + \xi_1) - k_p \varepsilon - k_d \epsilon_2 + v
\label{eq:pendulum_control}
\end{equation}
with gains $k_p,k_d>0$. Thus, the overall controller reads
\begin{equation}
    \begin{split}
        \dot{\xi}_1 =&\ \xi_2\\
        \dot{\xi}_2 =&\ -\omega_m^2 \sin(\varepsilon + \xi_1) + v
    \end{split}
\end{equation}
with control input $u$ defined by \eqref{eq:pendulum_control}.
As in the previous example the closed-loop system is the cascade of the system with state $(\varepsilon, \epsilon_2)$, whose origin is asymptotically stable, and a system (with state $(\xi_1, \xi_2)$) with vanishing input-bounded state properties. By Theorem~\ref{TH2} the problem is thus solved without necessarily forcing the trajectories of~\eqref{eq:pendulum_plant} to a predefined steady state. Simulations in Fig.~\ref{fig:epsilon_zero} confirm this. Fig.~\ref{fig:different_initial_conditions} shows that the actual steady-state reference trajectory reached by the output $\mu_1$ depends on the initial conditions of the controlled system. We further observe that the oscillation amplitude depends on the initial conditions of both the controlled system and the synchrony detector, and that there exists an initial condition of the reference model, shown in thick red (Fig.~\ref{fig:different_initial_conditions}), for which the resulting trajectories coincide with those of the undamped pendulum although with a different phase. As in the previous example, the output of the controlled system is not rigidly constrained by the initial state of either the detector or the reference model.
\begin{figure}
\begin{center}
\includegraphics[width=8.4cm, height=4.5cm]{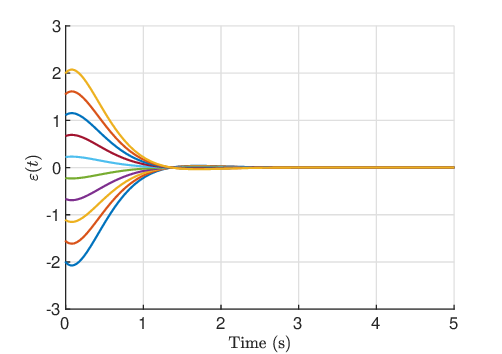}
\caption{Evolution of the synchronization error $\varepsilon(t)$ for different initial conditions of the controlled system. Results shown for $k_p = 10$ and $k_d = 5$.} 
\label{fig:epsilon_zero}
\end{center}
\end{figure}
\begin{figure}
\begin{center}
\includegraphics[width=8.4cm, height=4.5cm]{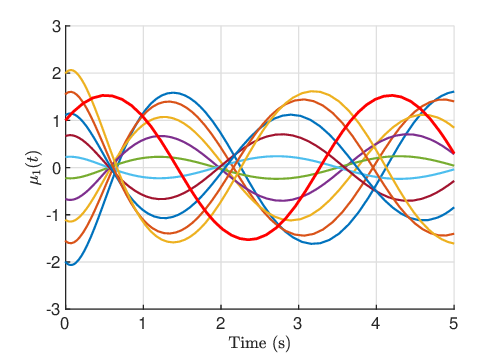}    
\caption{Evolution of the angle for different initial conditions, showing the ability to reproduce a family of behaviors. Results shown for $\omega_p = 10\,\mathrm{rad/s}$, $\omega_m =~2\,\mathrm{rad/s}$ and $v=0$. In red the trajectory for which the reference exhibits the same behavior of the controlled system, but with a noticeable phase shift.}
\label{fig:different_initial_conditions}
\end{center}
\end{figure}
\section{Concluding remarks}

The relationship between the constructions proposed in this paper and the various formulations of the internal model principle and model matching is subtle. In the classical nonlinear output regulation setting, as in \cite{ByrnesIsidori2003}, and in the asymptotic model matching framework of \cite{DiBenedettoGrizzle1994}, the analysis is carried out on an error system: controlled and reference dynamics are combined into an extended system, a zero-error manifold is identified, and necessary conditions for solvability are expressed in terms of invariance and stability properties of this manifold. On that manifold the closed-loop dynamics must embed the target dynamics, so that the controller necessarily contains a subsystem that generates the steady-state input keeping the error identically zero. In this sense, the internal model appears as a precise geometric object tied to the error and to its zero dynamics.

In the asymptotic behavior assignment formulation adopted here, we do not construct such an error system, nor do we embed a literal copy of the reference dynamics in the controller. Nevertheless, one should still expect some form of internal model to be present in the closed-loop system: along trajectories that achieve behavior assignment there must exist internal dynamics that reproduces the signals generated by the reference behavior. The crucial difference is that the specification we impose is weaker---we ask for asymptotic input--output behavior matching rather than full trajectory tracking of a designated reference solution. It is therefore reasonable to expect that any future necessary condition in this setting will also be weaker than those in \cite{DiBenedettoGrizzle1994}, reflecting the fact that input--output behavior enforcement is less demanding than asymptotic error convergence on an extended state space.

A further distinction concerns robustness. In the linear theory of \cite{FrancisWonham1976}, robust regulation with respect to perturbations of the controlled system's parameters requires that a physical regulation error is available and that the controller contains an internal model of the exosystem generating the exogenous signals. In our framework no such physical error is used: the synchrony detectors and the signals they produce are inherently model-based. It is therefore natural to regard the present results as nominal, and to anticipate that the flexibility gained by avoiding error feedback is paid for by a loss of robustness. Making this trade-off precise, and understanding to what extent internal-model-like structures can still be identified under suitable robustness assumptions in an asymptotic behavior assignment formulation, are topics for future work.

More broadly, the examples considered in this paper suggest that splitting the control objective into disturbance generation and behavior reproduction may lead to a refined internal model picture: one internal structure associated with the dynamics of exogenous signals and another associated with the desired input--output behavior to be enforced. Developing a geometric theory of such structures, in the spirit of \cite{DiBenedettoGrizzle1994} and \cite{ByrnesIsidori2003}, but tailored to open reference systems and to the absence of explicit error feedback, is part of an ongoing research programme.

\bibliography{ifacconf}             
                                            











\end{document}